\documentclass[letterpaper, 10 pt, conference]{ieeeconf}
\usepackage{cite,tikz}
\usepackage{amsmath,amssymb, amsfonts}
\usepackage{algorithmic}
\usepackage{graphicx}
\usepackage{epstopdf}
\usepackage{textcomp}
\usepackage{multirow}

\IEEEoverridecommandlockouts

\newtheorem{theorem}{Theorem}
\newtheorem{problem}{Problem}
\newtheorem{lemma}{Lemma}
\newtheorem{remark}{Remark}
\def\BibTeX{{\rm B\kern-.05em{\sc i\kern-.025em b}\kern-.08em
    T\kern-.1667em\lower.7ex\hbox{E}\kern-.125emX}}
\graphicspath{{Figure/}}
  \DeclareGraphicsExtensions{.pdf,.eps}

\begin{document}

\title{\LARGE \bf A Real-Time Optimal Eco-driving for
Autonomous Vehicles Crossing Multiple Signalized Intersections*}
\author{Xiangyu Meng$^{1}$ and Christos~G.~Cassandras$^{1}$\thanks{%
*This work was supported in part by NSF under grants ECCS-1509084,
IIP-1430145, and CNS-1645681, by AFOSR under grant FA9550-12-1-0113, by ARPA-E's NEXTCAR program under grant DE-AR0000796, and by Bosch and the MathWorks.}\thanks{$^{1}$The
authors are with the Division of Systems Engineering and Center for
Information and Systems Engineering, Boston University, Brookline, MA 02446,
USA \texttt{\small \{xymeng,cgc\}@bu.edu}}}
\maketitle

\begin{abstract}
This paper develops an optimal acceleration/speed
profile for a single autonomous vehicle crossing multiple signalized intersections
without stopping in free flow mode. The design objective is to produce both time and energy
efficient acceleration profiles of autonomous vehicles based on
vehicle-to-infrastructure communication. Our design approach differs from
most existing approaches based on numerical calculations: it
begins with identifying the structure of the optimal acceleration profile and then showing that it is characterized by several parameters, which are used for design
optimization. Therefore, the infinite dimensional optimal control problem is
transformed into a finite dimensional parametric optimization problem, which
enables a real-time online analytical solution. The simulation results show
quantitatively the advantages of considering multiple intersections jointly
rather than dealing with them individually. Based on mild assumptions, the optimal eco-driving algorithm is readily extended to include interfering traffic.
\end{abstract}


\thispagestyle{empty} \pagestyle{empty}


\section{Introduction}

The alarming state of existing transportation systems has been well
documented. For instance, in 2014, congestion caused vehicles in urban areas
to spend 6.9 billion additional hours on the road at a cost of an extra 3.1
billion gallons of fuel, resulting in a total cost estimated at 160 billion
\cite{schrank20152015}. From a control and optimization standpoint, the
challenges stem from requirements for increased safety, increased efficiency
in energy consumption, and lower congestion both in highway and urban
traffic. Connected and automated vehicles (CAVs), commonly known as
self-driving or autonomous vehicles, provide an intriguing opportunity for
enabling users to better monitor transportation network conditions and to
improve traffic flow. Their proliferation has rapidly grown, largely as a
result of Vehicle-to-X (or V2X) technology \cite{li2014survey} which refers
to an intelligent transportation system where all vehicles and
infrastructure components are interconnected with each other. Such
connectivity provides precise knowledge of the traffic situation across the
entire road network, which in turn helps optimize traffic flows, enhance
safety, reduce congestion, and minimize emissions. Controlling a vehicle to
improve energy consumption has been studied extensively, e.g., see \cite%
{gilbert1976vehicle, hooker1998optimal, hellstrom2010design, li2012minimum}.
By utilizing road topography information, an energy-optimal control
algorithm for heavy diesel trucks is developed in \cite{hellstrom2010design}%
. Based on Vehicle-to-Vehicle (V2V) communication, a minimum energy control
strategy is investigated in car-following scenarios in \cite{li2012minimum}.
Another important line of research focuses on coordinating vehicles at
intersections to increase traffic flow while also reducing energy
consumption. Depending on the control objectives, work in this area can be
classified as dynamically controlling traffic signals \cite{fleck2016adaptive}
and as coordinating vehicles \cite{milanes2010controller},\cite%
{alonso2011autonomous},\cite{huang2012assessing},\cite{kim2014mpc}. More
recently, an optimal control framework is proposed in \cite%
{malikopoulos2018decentralized} for CAVs to cross one or two adjacent
intersections in an urban area. The state of art and current trends in the
coordination of CAVs is provided in \cite{rios2017survey}.

Our focus in this paper is on an optimal control approach for a single
autonomous vehicle approaching multiple intersections in free flow mode in terms of energy
consumption and taking advantage of traffic light information. The term
\textquotedblleft ECO-AND\textquotedblright\ (short for \textquotedblleft
Economical Arrival and Departure\textquotedblright) is often used in the literature to refer
to this problem \cite{barth2011dynamic}. Its solution is made possible by vehicle-to-infrastructure
(V2I) communication, which enables a vehicle to automatically receive
signals from upcoming traffic lights before they appear in its visual range.
For example, such a V2I communication system has been launched in Audi cars
in Las Vegas by offering a traffic light timer on their dashboards: as the
car approaches an intersection, a red traffic light symbol and a
\textquotedblleft time-to-go\textquotedblright\ countdown appear in the
digital display and reads how long it will be before the traffic light ahead
turns green \cite{v2i}. Clearly, an autonomous vehicle can take advantage of
such information in order to go beyond current \textquotedblleft
stop-and-go\textquotedblright\ to achieve \textquotedblleft
stop-free\textquotedblright\ driving. Along these lines, the problem of
avoiding red traffic lights is investigated in \cite{asadi2011predictive,
mahler2014optimal, wan2016optimal, de2016eco}. The purpose in \cite%
{asadi2011predictive} is to track a target speed profile, which is generated
based on the feasibility of avoiding a sequence of red lights. The approach
uses model predictive control based on a receding horizon. Avoiding red
lights with probabilistic information at multiple intersections is
considered in \cite{mahler2014optimal}, where the time horizon is
discretized and deterministic dynamic programming is utilized to numerically
compute the optimal control input. The work in \cite{wan2016optimal} devises
the optimal speed profile given the feasible target time, which is within
some green light interval. A velocity pruning algorithm is proposed in~\cite%
{de2016eco} to identify feasible green windows, and a velocity profile is
calculated numerically in terms of energy consumption. Most existing work
solves the eco-driving problem with traffic light constraints numerically,
such as using dynamic programming \cite{mahler2014optimal, sun2018robust}, and
model predictive control \cite{asadi2011predictive}. To enable the real-time
use of such eco-driving methods, it is desirable to have an on-line
analytical solution.

From the above, it is clear that a need exists for developing new methods
for eco-driving of autonomous vehicles with traffic light constraints. This
paper aims to address this need by proposing an extension to our previous
approach from a single signalized intersection \cite{meng2018optimal, meng2018optimala} to
multiple intersections. We show explicitly that the optimal acceleration profile has a piecewise linear form, similar to the results in \cite{meng2018optimal, meng2018optimala}, that includes all state equality and temporal inequality constraints involved. It follows from the theoretical analysis that the optimal acceleration profile can be parameterized by a piecewise linear function of time, which offers a real-time analytical solution to eco-driving of
autonomous vehicles crossing multiple signalized intersections without
stopping. We illustrate the effectiveness of the proposed optimal parametric approach
through simulations and show that it yields better results
compared with our previous eco-driving approach \cite{meng2018optimal, meng2018optimala} applied
to each intersection individually. We also show that the optimal eco-driving algorithm can be adjusted to handle the case with interfering traffic under the assumption of the availability of some traffic information.

\section{Problem Formulation}

The vehicle dynamics are modeled by a double integrator%
\begin{align}
\dot{x}\left( t\right) &=v\left( t\right) ,  \label{dd} \\
\dot{v}\left( t\right) &=u\left( t\right) ,  \label{vd}
\end{align}%
where $x\left( t\right) $ is the travel distance of the vehicle relative to some origin on the road, which may include turns, $v\left( t\right) $
the velocity, and $u\left( t\right) $ the acceleration/deceleration. At $%
t_{0}$, the initial position and velocity are given by $x\left( t_{0}\right)
=x_{0}$ and $v\left( t_{0}\right) =v_{0}$, respectively. On-road vehicles
have to obey traffic rules, such as the minimum and maximum speed permitted $0<v_{\min }\leq v\left( t\right) \leq v_{\max }$.
The physical constraints on acceleration and deceleration are determined by
vehicle parameters $u_{\min }\leq u\left( t\right) \leq u_{\max }$,
where $u_{\min }<0$ and $u_{\max }>0$ denote the maximum deceleration and
acceleration, respectively.

Assume that there are $N$ intersections. Each intersection $i$ is equipped with a traffic
light, which is dictated by the square wave $f_{i}\left( t\right) $ defined below%
\begin{equation*}
f_{i}\left( t\right) =\left\{
\begin{array}{ll}
1, & \text{when }kT_{i}\leq t\leq kT_{i}+D_{i}T_{i}, \\
0, & \text{when }kT_{i}+D_{i}T_{i}<t<kT_{i}+T_{i},%
\end{array}%
\right.
\end{equation*}%
where $f_{i}\left( t\right) =1$ indicates that the traffic light is green,
and $f_{i}\left( t\right) =0$ indicates that the traffic light is red. The
parameter $0<D_{i}<1$ is the fraction of the time period $T_{i}$ during
which the traffic light is green, and $k=0,1,\ldots ,$ are non-negative integers. Assume that there is no offset among the signals. Our algorithm also supports dynamically actuated traffic signals if the time until green/red can be determined and communicated to the autonomous vehicle. Then we can re-solve the problem with the new timing information.

Let $\{t_{i}\}_{i=1}^{N}$ be a sequence of intersection crossing times with $%
t_{i+1}>t_{i}$, that is, $x\left( t_{i}\right) =\sum\nolimits_{j=1}^{i}l_{j}$,
where $l_{j}$ is the length of road segment $j$. To ensure stop-free
intersection crossing, $t_{i}$ must be within the green light interval, that
is, $kT_{i}\leq t_{i}\leq kT_{i}+D_{i}T_{i}$ for some non-negative integer $k$.

Our objective is the eco-driving of autonomous vehicles crossing multiple
intersections in terms of both time and energy efficiency. Therefore, our
problem formulation is given below:

\begin{problem}
\label{P1} ECO-AND Problem%
\begin{equation*}
\min_{u\left( t\right) }\rho _{t}\left( t_{p}-t_{0}\right) +\rho
_{u}\int_{t_{0}}^{t_{p}}u^{2}\left( t\right) dt
\end{equation*}%
subject to%
\begin{align}
&\left( \text{\ref{dd}}\right) \text{ and }\left( \text{\ref{vd}}\right)
\label{c1} \\
&x\left( t_{i}\right) =\sum\nolimits_{j=1}^{i}l_{j},\text{ }i=1,\ldots ,N
\label{c2} \\
&v_{\min }\leq v\left( t\right) \leq v_{\max }  \label{c3} \\
&u_{\min }\leq u\left( t\right) \leq u_{\max }  \label{c4} \\
&k_{i}T_{i}\leq t_{i}\leq k_{i}T_{i}+D_{i}T_{i},\text{ }i=1,\ldots ,N
\label{c5}
\end{align}%
for some non-negative integers $k_{i}$, where $\rho _{t}$ and $\rho _{u}$
are the weighting parameters, and $t_{p}=t_{N}$ is the time when the vehicle
arrives at the last intersection.
\end{problem}

Procedures for
normalizing these two terms for the purpose of a well-defined optimization
problem can be found in~\cite{meng2018optimal, meng2018optimala} by properly determining
weights $\rho _{t}$ and $\rho _{u}$.
In Problem~\ref{P1}, the term $J^{t}=t_{p}-t_{0}$ is the travel time while $%
J^{u}=\int_{t_{0}}^{t_{p}}u^{2}\left( t\right) dt$ captures the energy
consumption; see \cite{malikopoulos2013optimization}.

\section{Main Results}

Before proceeding further, let us first introduce a lemma, which will be
used frequently throughout the following analysis.

\begin{lemma}
\label{l1} Consider the vehicle's dynamics (\ref{dd}) and (\ref{vd}) with
initial conditions $x_{0}$ and $v_{0}$. If the acceleration profile of the
vehicle has the form $u\left( t\right) =at+b$
during the time interval $\left[ t_{0},t_{1}\right] ,$ where $a$ and $b$ are
two constants, then%
\begin{align*}
v\left( t_{1}\right) =&v_{0}+b\left( t_{1}-t_{0}\right) +\frac{a}{2}\left(
t_{1}^{2}-t_{0}^{2}\right) \text{,} \\
x\left( t_{1}\right) =&x_{0}+v_{0}\left( t_{1}-t_{0}\right) +\frac{1}{2}%
b\left( t_{1}-t_{0}\right) ^{2} \\
&\quad+\frac{a}{6}\left( t_{1}^{3}+2t_{0}^{3}-3t_{0}^{2}t_{1}\right) \text{,}
\\
J^{u}\left( t_{1}\right) =&\frac{a^{2}}{3}\left( t_{1}^{3}-t_{0}^{3}\right)
+ab\left( t_{1}^{2}-t_{0}^{2}\right) +b^{2}\left( t_{1}-t_{0}\right) \text{.}
\end{align*}
\end{lemma}
The proof is obtained by using the kinematic equations of the
vehicle (\ref{dd}) and (\ref{vd}) and the definition of $J_{u}$. Due to space constraints, the details are omitted.
\begin{remark}
We will show in Theorem~\ref{T1} below that in fact the \emph{optimal} acceleration
profile for Problem~\ref{P1} is of the form $u\left( t\right) =at+b$, which
captures most acceleration profiles used in the literature and vehicle
simulation software \cite{prescan2016}. When $a=b=0$, the vehicle travels at
a constant speed. When $a=0$, the acceleration profile becomes either
constant acceleration ($b>0$) or constant deceleration ($b<0$). When $a\neq 0
$, the resulting linear acceleration profile is also called ``smooth jerk"
\cite{prescan2016}.
\end{remark}

In order not to overshadow the main idea, we consider the case of only two
consecutive intersections here, where $t_{p}=t_{2}$. We will show how the proposed framework can
include our previous result on a single intersection \cite{meng2018optimal, meng2018optimala}
as a special case in Subsection~\ref{fds} and can be extended to the case of
more than two intersections in Subsection~\ref{fdm}.

The main challenge of extending the result from one intersection \cite%
{meng2018optimal, meng2018optimala} to multiple intersections lies in the interior-point
constraints $x\left( t_{1}\right) =l_{1}$ and $kT_{1}\leq t_{1}\leq kT_{1}+D_{1}T_{1}$. Note that we have both a spatial equality constraint and a temporal
inequality constraint. Other constraints, such as states,
acceleration/deceleration, and terminal constraints, have been thoroughly
studied in our previous work \cite{meng2018optimal, meng2018optimala}. The following theorem
shows how the optimal acceleration profile is affected by the interior-point
constraints.

\begin{theorem}
\label{T1} The optimal acceleration profile $u^{\ast }\left( t\right) $ of
Problem \ref{P1} has the form%
\begin{equation*}
u^{\ast }\left( t\right) =a\left( t\right) t+b\left( t\right) ,
\end{equation*}%
where $a\left( t\right) $ and $b\left( t\right) $ are piece-wise constant.
In addition, $u^{\ast }\left( t\right) $ is continuous everywhere, and $%
u^{*}(t_{p}^{*})=0$.
\end{theorem}

\begin{proof}
The interior-point constraints are dealt with by using the calculus of variations methodology
borrowed from \cite{bryson1975applied} with certain modifications. The
Hamiltonian $H(v,u,\lambda)$ and Lagrangian $L(v,u,\lambda,\mu,\eta)$ are defined as%
\begin{equation*}
H\left( v,u,\lambda\right) =\rho _{t}+\rho _{u}u^{2}\left( t\right) +\lambda
_{1}\left( t\right) v\left( t\right) +\lambda _{2}\left( t\right) u\left(
t\right)
\end{equation*}%
and%
\begin{align*}
L\left( v,u,\lambda,\mu,\eta\right)  =&H\left( v,u,\lambda\right) +\eta _{1}\left( t\right) \left[
v_{\min }-v\left( t\right) \right]  \\
&+\eta _{2}\left( t\right) \left[ v\left(
t\right) -v_{\max }\right] \\
&+\mu _{1}\left( t\right) \left[ u_{\min }-u\left( t\right) \right]\\
& +\mu_{2}\left( t\right) \left[ u\left( t\right) -u_{\max }\right] \text{,}
\end{align*}%
respectively, where $\lambda \left( t\right) =\left[ \lambda _{1}\left( t\right) \text{ }%
\lambda _{2}\left( t\right) \right] ^{T}$, $\mu \left( t\right) =\left[ \mu _{1}\left( t\right) \text{ }%
\mu _{2}\left( t\right) \right] ^{T}$, $\eta \left( t\right) =\left[ \eta _{1}\left( t\right) \text{ }%
\eta _{2}\left( t\right) \right] ^{T}$, and%
\begin{align}
&\eta _{1}\left( t\right) \geq 0,\text{ }\eta _{2}\left( t\right) \geq 0,
\label{xc1} \\
&\eta _{1}\left( t\right) \left[ v_{\min }-v\left( t\right) \right] +\eta
_{2}\left( t\right) \left[ v\left( t\right) -v_{\max }\right] =0,
\label{xc2} \\
&\mu _{1}\left( t\right) \geq 0,\text{ }\mu _{2}\left( t\right) \geq 0,
\label{xc3} \\
&\mu _{1}\left( t\right) \left[ u_{\min }-u\left( t\right) \right] +\mu
_{2}\left( t\right) \left[ u\left( t\right) -u_{\max }\right] =0.
\label{xc4}
\end{align}

According to Pontryagin's minimum principle, the optimal control $u^{\ast
}\left( t\right) $ must satisfy%
\begin{equation}
u^{\ast }\left( t\right) =\arg \min_{u_{\min }\leq u\left( t\right) \leq
u_{\max }}H\left( v^{\ast }\left( t\right) ,u\left( t\right) ,\lambda \left(
t\right) \right) ,  \label{pya}
\end{equation}%
which allows us to express $u^{\ast }\left( t\right) $ in terms of the
co-state $\lambda \left( t\right) ,$ resulting in%
\begin{equation}
u^{\ast }\left( t\right) =\left\{
\begin{array}{ll}
u_{\max } & \text{when }-\frac{\lambda _{2}\left( t\right) }{2\rho _{u}}\geq
u_{\max } \\
-\frac{\lambda _{2}\left( t\right) }{2\rho _{u}} & \text{when }u_{\min }\leq
-\frac{\lambda _{2}\left( t\right) }{2\rho _{u}}\leq u_{\max } \\
u_{\min } & \text{when }-\frac{\lambda _{2}\left( t\right) }{2\rho _{u}}\leq
u_{\min }%
\end{array}%
\right.   \label{optc}
\end{equation}%
For simplicity, we write $L(t)$, $H(t)$, $J(t)$ and $N(t)$ without the arguments of states, co-states, and multipliers in the rest of the paper. Adjoin the system differential equations (\ref{dd}) and (\ref{vd}) to $%
L\left( t\right) $ with multiplier function $\lambda \left( t\right) $:%
\begin{equation*}
J\left( t\right) =N\left( t_{1},t_{p}\right) +\int_{t_{0}}^{t_{p}}\left[ L\left(
t\right) -\lambda _{1}\left( t\right) \dot{x}\left( t\right) -\lambda
_{2}\left( t\right) \dot{v}\left( t\right) \right] dt,
\end{equation*}%
where $t_{1}$ is the first intersection crossing time, and%
\begin{align*}
&N\left( t_{1},t_{p}\right)  =\nu x\left( t_{p}\right) +\pi \left( x\left(
t_{1}\right) -l\right) +\upsilon _{1}\left( k_{1}T_{1}-t_{1}\right)  \\
&\qquad \qquad+\upsilon _{2}\left( t_{1}-k_{1}T_{1}-D_{1}T_{1}\right),\\
&\upsilon _{1}\geq 0, \quad \upsilon _{2}\geq 0,\\
&\upsilon _{1}\left( k_{1}T_{1}-t_{1}\right)+\upsilon _{2}\left( t_{1}-k_{1}T_{1}-D_{1}T_{1}\right)=0.
\end{align*}%
The first variation of the augmented performance index is%
\begin{align*}
\delta J\left( t\right) =&\delta N\left( t_{1}, t_{p}\right) \\
&+\delta
\int_{t_{0}}^{t_{p}}\left[ L\left( t\right) -\lambda _{1}\left( t\right)
\dot{x}\left( t\right) -\lambda _{2}\left( t\right) \dot{v}\left( t\right) %
\right] dt\text{.}
\end{align*}%
Split the integral into two parts:%
\begin{eqnarray*}
\delta J\left( t\right)  &=&\nu \delta x\left( t\right)
|_{t=t_{p}}+\left( \upsilon _{2}-\upsilon _{1}\right) dt_{1}+\pi dx\left(
t_{1}\right)  \\
&&+\left[ L\left( t\right) -\lambda _{1}\left( t\right) \dot{x}\left(
t\right) -\lambda _{2}\left( t\right) \dot{v}\left( t\right) \right]
|_{t=t_{1}^{-}}dt_{1} \\
&&-\left[ L\left( t\right) -\lambda _{1}\left( t\right) \dot{x}\left(
t\right) -\lambda _{2}\left( t\right) \dot{v}\left( t\right) \right]
|_{t=t_{1}^{+}}dt_{1} \\
&&-\lambda _{1}\left( t\right) \delta x\left( t\right)
|_{t_{0}}^{t_{1}^{-}}-\lambda _{1}\left( t\right) \delta x\left( t\right)
|_{t_{1}^{+}}^{t_{p}} \\
&&-\lambda _{2}\left( t\right) \delta v\left( t\right)
|_{t_{0}}^{t_{1}^{-}}-\lambda _{2}\left( t\right) \delta v\left( t\right)
|_{t_{1}^{+}}^{t_{p}} \\
&&+\int_{t_{0}}^{t_{p}}\left[ \dot{\lambda}_{2}\left( t\right) +\frac{%
\partial L\left( t\right) }{\partial v\left( t\right) }\right] \delta
v\left( t\right) dt \\
&&+\int_{t_{0}}^{t_{p}}\left\{ \dot{\lambda}_{1}\left( t\right) \delta
x\left( t\right) +\frac{\partial L\left( t\right) }{\partial u\left(
t\right) }\delta u\left( t\right) \right\} dt,
\end{eqnarray*}%
where we let $t_{1}^{-}$ signify just before $t_{1}$ and $t_{1}^{+}$ signify just after $t_{1}$.
We now make use of the relationships%
\begin{equation*}
dx\left( t_{1}\right) =\left\{
\begin{array}{c}
\delta x\left( t_{1}^{-}\right) +\dot{x}\left( t_{1}^{-}\right) dt_{1}, \\
\delta x\left( t_{1}^{+}\right) +\dot{x}\left( t_{1}^{+}\right) dt_{1},%
\end{array}%
\right.
\end{equation*}%
and the relationships of $v\left( t\right) $ can be derived similarly. Using
the above relationships to eliminate $\delta x\left( t_{1}^{-}\right) $ and $%
\delta x\left( t_{1}^{+}\right) $, and regrouping terms, yields%
\begin{align}
\delta J\left( t\right)  =&\left[ v-\lambda _{1}\left( t\right) %
\right] \delta x\left( t\right) |_{t=t_{p}}-\lambda _{2}\left( t\right)
\delta v\left( t\right) |_{t=t_{p}} \notag\\
&+\lambda _{1}\left( t\right) \delta x\left( t\right) |_{t=t_{0}}+\lambda
_{2}\left( t\right) \delta v\left( t\right) |_{t=t_{0}} \notag\\
&+\left[ L\left( t_{1}^{-}\right) -L\left( t_{1}^{+}\right) +\upsilon
_{2}-\upsilon _{1}\right] dt_{1} \label{VA} \\
&+\left[ \lambda _{1}\left( t_{1}^{+}\right) -\lambda _{1}\left( t_{1}^{-}\right)
+\pi \right] dx\left( t_{1}\right)  \notag\\
&+\left[ \lambda _{2}\left( t_{1}^{+}\right) -\lambda _{2}\left( t_{1}^{-}\right) %
\right] dv\left( t_{1}\right)\notag
\end{align}%
Since we have no constraint on $v\left( t\right) $ at $t=t_{1}$, it follows that $\lambda _{2}\left( t_{1}^{+}\right) =\lambda _{2}\left( t_{1}^{-}\right)$,
that is to say, there are no discontinuities in $\lambda _{2}\left( t\right)
$ at $t=t_{1}$. Therefore, $u^{\ast }\left( t\right) $ is continuous
everywhere based on (\ref{pya}) and Theorem 1 in \cite{meng2018optimal, meng2018optimala}. To make the term $\lambda_2(t_{p})\delta v(t_{p})$ in (\ref{VA}) vanish, we must have $\lambda_2(t_{p})=0$ since there are no constraints on $v (t)$ at $t = t_{p}$. From the optimality condition (\ref{optc}), we have $u^{*}(t_{p}^{*})=0$.

For the co-state $\lambda _{1}\left( t\right) ,$ we have%
\begin{equation*}
\dot{\lambda}_{1}\left( t\right) =-\frac{\partial L^{\ast }\left( t\right) }{%
\partial x}=0.
\end{equation*}%
However, since $dx\left( t_{1}\right) =0,$ $\lambda _{1}\left( t\right) $ may or may not have jumps at $t=t_{1}$. Therefore, $%
\lambda _{1}\left( t\right) $ can be written as%
\begin{equation}
\lambda _{1}\left( t\right) =\left\{
\begin{array}{cc}
\lambda _{1}^{-} & \text{for }t_{0}\leq t\leq t_{1}, \\
\lambda _{1}^{+} & \text{for }t_{1}<t\leq t_{p}\text{.}%
\end{array}%
\right.  \label{lamb1}
\end{equation}%
For the co-state $\lambda _{2}\left( t\right) ,$ we have%
\begin{equation}
\dot{\lambda}_{2}\left( t\right) =-\frac{\partial L^{\ast }\left( t\right) }{%
\partial v}=-\lambda _{1}\left( t\right) +\eta _{1}\left( t\right) -\eta
_{2}\left( t\right) \text{.}  \label{lamb2}
\end{equation}%
Depending on the value of $v\left( t\right) ,$ we have different cases:

Case I: $v_{\min }<v\left( t\right) <v_{\max }$. In this case, $\eta
_{1}\left( t\right) =\eta _{2}\left( t\right) =0$. Therefore, $\lambda
_{2}\left( t\right) $ linearly increases or decreases according to (\ref%
{lamb1}) and (\ref{lamb2}), and so does $u^{\ast }\left( t\right) $ based on (\ref%
{optc}).

Case II: $v\left( \tau\right) =v_{\min }$. In this case, we have $u^{\ast }\left(
t\right) \geq 0$ over some interval $[\tau, \tau+\alpha]$. When $u^{\ast }\left( t\right) =0$ over the interval $[\tau, \tau+\alpha]$, we must have $\lambda_{2}(t)=\dot{\lambda}_{2}(t)=0$, that is, $\eta
_{1}\left( t\right) =\lambda _{1}\left( t\right) $ from (\ref{lamb2}) and the fact that $\eta_{2}(t)=0$ based on (\ref{xc2}). When $u^{\ast }\left(
\tau^{+}\right) >0,$ $v\left( \tau^{+}\right) >v_{\min }$. Then, it becomes Case I.

Case III: $v\left( \tau\right) =v_{\max }$. In this case, we have $u^{\ast }\left(
t\right) \leq 0$ over some interval $[\tau, \tau+\alpha]$. When $u^{\ast }\left( t\right) =0$ over the interval $[\tau, \tau+\alpha]$, we have $\lambda_{2}(t)=\dot{\lambda}_{2}(t)=0$, that is, $\eta
_{2}\left( t\right) +\lambda _{1}\left( t\right) =0$ from (\ref{lamb2}) and the fact that $\eta_{1}(t)=0$ based on (\ref{xc2}). When $u^{\ast }\left(
\tau^{+}\right) <0$, $v\left( \tau^{+}\right) <v_{\max }$. Then, it becomes Case I.

Regardless of which of these three cases applies, the optimal control $u^{\ast }\left( t\right) $
always has a linear form.
\end{proof}

\begin{remark}
Assume that at $t_{1}$ all the states and acceleration/deceleration
constraints are relaxed. Then $L\left( t\right) $ is the same as $H\left(
t\right) $. To cause the coefficient of $dt_{1}$ to vanish, the condition $L\left( t_{1}^{-}\right) -L\left( t_{1}^{+}\right) +\upsilon _{2}-\upsilon
_{1}=0$ has to be satisfied. If $kT_{1}<t_{1}<kT_{1}+D_{1}T_{1}$, then $\upsilon
_{2}=\upsilon _{1}=0$. Therefore, there are no jumps in $L\left( t\right) $
and $H\left( t\right) $ at $t_{1}$. In other words, the co-state $\lambda_{1}
$ has no jumps in this case. However, when $t_{1}=kT_{1}$ or $%
t_{1}=kT_{1}+D_{1}T_{1}$, there may be jumps in $L\left( t\right) $ and $%
H\left( t\right) $ at $t_{1}$. Then $\lambda_{1}$ switches from one value to
another as shown in (\ref{lamb1}).
\end{remark}

Based on Theorem~\ref{T1}, we know that the optimal acceleration profile has
the form $u^{\ast }\left( t\right) =a\left( t\right) t+b\left( t\right)$, where $a\left( t\right) $ and $b\left( t\right) $ are piece-wise constant. \
For example, we have $a\left( t\right) =0, \quad b\left( t\right) =u_{\max }$ for $u^{\ast }\left( t\right) =u_{\max }$, and $a\left( t\right) =0, \quad b\left( t\right) =u_{\min }$ for $u^{\ast }\left( t\right) =u_{\min }$. For the case that $u\left(
t\right) =0$, we could set $a\left( t\right)=b\left( t\right) =0$. Most of
the time, $a\left( t\right) =a$ and $b\left( t\right) =b$ are just
constants. In addition, there are only a few time instants when $a\left(
t\right) $ and $b\left( t\right) $ switch from a constant to another. Such
instants include the time when the maximum acceleration starts to decrease,
the maximum deceleration starts to increase, the vehicle reaches the maximum
or the minimum allowed speed limits, or at the first intersection crossing time $t_{1}
$. Therefore, we could parameterize the \emph{optimal} acceleration profile by a
sequence of linear functions of time.

\section{Parametric Optimization}
Based on the analysis of the last section, the optimal acceleration profile can be parameterized by a sequence of linear functions of time, such as the one shown in Fig.~\ref{fig3f}.
Now let us split the analysis into two parts: $\left[ t_{0},t_{1}\right] $
and $\left[ t_{1},t_{2}\right] $, where $t_{1}$ is the first intersection
crossing time. The optimal acceleration profile at most has four switches at
$\tau_{1}, \tau_{3}, \tau_{5}, \tau_{6} $ as shown in Fig.~\ref{fig3f}. The
acceleration profile shown in Fig.~\ref{fig3f} is the most
complicated acceleration profile possible, starting with the maximum acceleration,
which can be obtained based on the optimality condition (\ref{optc}) and the
following facts:%
\begin{itemize}
\item $u^{\ast }\left( t_{2}^{*}\right) =0,$ which can be seen from Theorem~%
\ref{T1}.

\item Whenever $v\left( t\right) =v_{\min }$ or $v_{\max }$, $u^{\ast
}\left( t\right) =0$.

\item $u^{\ast }\left( t\right) $ is continuous without jumps.

\item Only at $t_{1},$ $\dot{\lambda}_{2}\left( t\right) $ may change sign
according to (\ref{lamb2}).
\end{itemize}
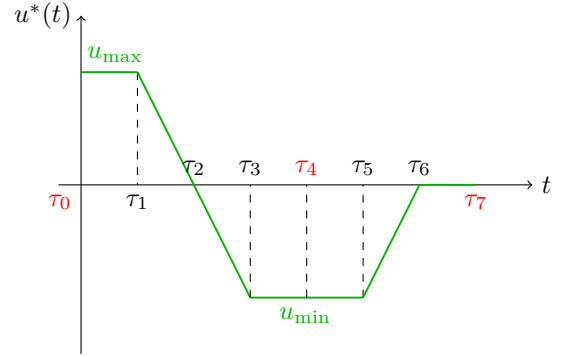
\begin{figure}[htbp]
\centering
\begin{tikzpicture}[scale=1.5]
  \draw[->] (-0.2,0) -- (4,0) node[right] {$t$};
  \draw[->] (0,-1.5) -- (0,1.5) node[left] {$u^{*}(t)$};
  \draw[line width=0.25mm,color=black!30!green](0,1) -- (0.5,1) node[above,xshift=-0.3cm] {$u_{\max}$};
  \draw[dashed] (0.5,1) -- (0.5,0) node[below] {$\tau_{1}$};
  \draw[line width=0.25mm,color=black!30!green](0.5,1) -- (1.5,-1);
  \draw[dashed] (1.5,-1) -- (1.5,0) node[above] {$\tau_{3}$};
  \draw[line width=0.25mm,color=black!30!green](1.5,-1) -- (2.5,-1) node[below left,xshift=-0.3cm] {$u_{\min}$};
  \draw[dashed] (2,-1) -- (2,0) node[above,red] {$\tau_{4}$};
  \draw[dashed] (2.5,-1) -- (2.5,0) node[above] {$\tau_{5}$};
  \draw[line width=0.25mm,color=black!30!green](2.5,-1) -- (3,0);
  \draw[line width=0.25mm,color=black!30!green](3,0) -- (3.5,0);
  \draw(0,0)node[below left,red] {$\tau_{0}$};
  \draw(1,0)node[above] {$\tau_{2}$};
  \draw(3,0)node[above] {$\tau_{6}$};
  \draw[red](3.5,0)node[below] {$\tau_{7}$};
\end{tikzpicture}
\caption{Optimal acceleration profile for two intersections}
\label{fig3f}
\end{figure}
A similar optimal acceleration profile can be drawn when it starts with the
maximum deceleration. The second fact above corresponds to the
interval $[\tau_{6},\tau_{7}]$ in Fig.~\ref{fig3f}, in which $v(t)=v_{\min}$
for $t\in [\tau_{6},\tau_{7}]$. The fourth fact above can be
visualized in Fig.~\ref{fig3f} as well. Before $\tau_{4}$, the acceleration
decreases monotonically; and after $\tau_{4}$ it increases monotonically.
Even though there are five linear functions in Fig.~\ref{fig3f}, seven
linear functions are needed to parameterize the acceleration profile. Over
the interval $[\tau_{1}, \tau_{3}]$ in Fig.~\ref{fig3f}, there is only one
linear function. In order to guarantee that the acceleration profile during
each interval contains either acceleration or deceleration but not both,
which is ensured by the constraint (\ref{p22}) below, we consider that there
is a switch at $\tau_{2}$ between acceleration and deceleration. Therefore,
two linear functions are used to parameterize the optimal acceleration
profile. By doing so, the speed either increases or decreases during each
interval. Therefore, the constraint (\ref{p21}) below ensures that the speed
is within the minimum and maximum bounds all the time.

We can thus parameterize the optimal acceleration profile as follows:%
\begin{equation*}
u^{\ast }\left( t\right) =\left\{
\begin{array}{cc}
a_{1}t+b_{1} & \text{for }t\in \left[ \tau _{0},\tau _{1}\right] \\
a_{2}t+b_{2} & \text{for }t\in \left[ \tau _{1},\tau _{2}\right] \\
a_{3}t+b_{3} & \text{for }t\in \left[ \tau _{2},\tau _{3}\right] \\
a_{4}t+b_{4} & \text{for }t\in \left[ \tau _{3},\tau _{4}\right] \\
a_{5}t+b_{5} & \text{for }t\in \left[ \tau _{4},\tau _{5}\right] \\
a_{6}t+b_{6} & \text{for }t\in \left[ \tau _{5},\tau _{6}\right] \\
a_{7}t+b_{7} & \text{for }t\in \left[ \tau _{6},\tau _{7}\right]%
\end{array}%
\right.
\end{equation*}%
where $\tau _{0}=t_{0},$ $\tau _{4}=t_{1}$, and $\tau _{7}=t_{2}$.

\begin{remark}
The optimal acceleration profile is overapproximated by the triplets $%
(a_{i},b_{i},\tau_{i})$, $i=1,2,\ldots,7$, 21 variables in total. The number
of variables can be reduced when the properties of $u^{*}(t)$ are
considered. The advantage of the parametric form of the optimal controller is that it simplifies the
complicated analysis through a computationally efficient scheme suitable for
real-time implementation. Also note that vehicles may experience both
acceleration and deceleration during a single road segment, which is
different from the optimal acceleration profile for a single intersection
\cite{meng2018optimal, meng2018optimala}.
\end{remark}

We have now shown that Problem~\ref{P1} is equivalent to this parametric optimization problem:

\begin{problem}
\label{P2} ECO-AND problem%
\begin{equation*}
\min \rho _{t}\tau_{7}+\rho _{u}\sum_{i=1}^{7}J^{u}_{i}
\end{equation*}%
subject to%
\begin{align}
&v_{\min }\leq v\left( \tau _{i}\right) \leq v_{\max },  \label{p21} \\
&\left( a_{i}\tau _{i}+b_{i}\right) \left( a_{i}\tau _{i-1}+b_{i}\right)
\geq 0,  \label{p22} \\
&u_{\min }\leq a_{i}\tau _{i}+b_{i}\leq u_{\max },  \label{p23} \\
&u_{\min }\leq a_{i}\tau _{i-1}+b_{i}\leq u_{\max },  \label{p24} \\
&\tau _{i-1}\leq \tau _{i},  \label{p25} \\
&i=1,\ldots ,7,  \notag \\
&k_{1}T_{1}\leq \tau _{4}\leq k_{1}T_{1}+D_{1}T_{1},  \label{p26} \\
&x\left( \tau _{4}\right) =l_{1}  \label{p27} \\
&k_{2}T_{2}\leq \tau _{7}\leq k_{2}T_{2}+D_{2}T_{2},  \label{p28} \\
&x\left( \tau _{7}\right) =l_{1}+l_{2}  \label{p29}
\end{align}%
where $J^{u}_{i}$ is the energy cost during the interval $%
[\tau_{i-1},\tau_{i}]$, which can be obtained as%
\begin{equation*}
J^{u}_{i}=\frac{a_{i}^{2}}{3}\left( \tau _{i}^{3}-\tau _{i-1}^{3}\right)
+a_{i}b_{i}\left( \tau _{i}^{2}-\tau _{i-1}^{2}\right) +b_{i}^{2}\left( \tau
_{i}-\tau _{i-1}\right)
\end{equation*}%
from Lemma~\ref{l1},
\begin{equation*}
v\left( \tau _{i}\right) =v\left( \tau _{i-1}\right) +b_{i}\left( \tau
_{i}-\tau _{i-1}\right) +\frac{a_{i}}{2}\left( \tau _{i}^{2}-\tau
_{i-1}^{2}\right)
\end{equation*}%
and%
\begin{eqnarray*}
x\left( \tau _{i}\right) &=&x\left( \tau _{i-1}\right) +v\left( \tau
_{i-1}\right) \left( \tau _{i}-\tau _{i-1}\right) \\
&&+\frac{b_{i}}{2}\left( \tau _{i}-\tau _{i-1}\right) ^{2}+\frac{a_{i}}{6}%
\left( \tau _{i}^{3}+2\tau _{i-1}^{3}-3\tau _{i-1}^{2}\tau _{i}\right) .
\end{eqnarray*}
\end{problem}

\begin{remark}
Problem~\ref{P2} is equivalent to Problem~\ref{P1}, where the continuous
velocity constraint (\ref{c3}) is ensured by (\ref{p21}) and (\ref{p22}).
The continuous acceleration constraint (\ref{c4}) is ensured by (\ref{p23}%
) and (\ref{p24}). The constraint (\ref{p25}) is needed to ensure the right order of the intersection crossing times of the autonomous vehicle driven by
the optimal acceleration profile.
\end{remark}

\begin{remark}
The parametric optimization framework is very general so that it can be used
to solve many different eco-driving problems. By taking into consideration
the driving comfort, we can just add the constraints $\left\vert
a_{i}\right\vert \leq a_{J}$, where $a_{i}$ corresponds to the jerk profile,
and $a_{J}$ is the limit of jerk tolerance \cite{barth2011dynamic}. The
parametric optimization framework can also easily incorporate an initial
acceleration condition, interior and terminal velocity/acceleration
constraints by adding additional equality or inequality constraints.
\end{remark}

In the following, we will show how this optimal parametric framework
includes our previous result \cite{meng2018optimal, meng2018optimala} as a special case, and
how to extend the framework to more than two intersections.

\subsection{Single Intersection}

\label{fds} Based on our analysis for a single intersection \cite%
{meng2018optimal, meng2018optimala}, the optimal acceleration profile can be parameterized as%
\begin{equation*}
u^{\ast }\left( t\right) =\left\{
\begin{array}{cc}
a_{1}t+b_{1} & \text{for }t\in \left[ \tau _{0},\tau _{1}\right] \\
a_{2}t+b_{2} & \text{for }t\in \left[ \tau _{1},\tau _{2}\right] \\
a_{3}t+b_{3} & \text{for }t\in \left[ \tau _{2},\tau _{3}\right]%
\end{array}%
\right.
\end{equation*}%
where $\tau _{0}=t_{0}$, and $\tau _{3}=t_{1}$. The optimal parameters $%
\left( a_{i},b_{i},\tau _{i}\right) $ for $i=1,2,3$ can be obtained by
solving the following optimization problem:

\begin{problem}
\label{P3} ECO-AND problem%
\begin{equation*}
\min \rho _{t}\tau_{3}+\rho _{u}\sum_{i=1}^{3}J^{u}_{i}
\end{equation*}%
subject to%
\begin{align}
&v_{\min }\leq v\left( \tau _{3}\right) \leq v_{\max }  \label{p31} \\
&u_{\min }\leq a_{1}\tau _{0}+b_{1}\leq u_{\max }  \label{p32} \\
&\tau _{i-1}\leq \tau _{i},\quad i=1,\ldots ,3  \label{p33} \\
&kT\leq \tau _{3}\leq kT+DT  \label{p34} \\
&x\left( \tau _{3}\right) =l,  \label{p35}
\end{align}%
where $J^{u}_{i}$ is the energy cost during the interval $\left[ \tau
_{i-1},\tau _{i}\right] ,$ which can be obtained as%
\begin{equation*}
J^{u}_{i}=\frac{a_{i}^{2}}{3}\left( \tau _{i}^{3}-\tau _{i-1}^{3}\right)
+a_{i}b_{i}\left( \tau _{i}^{2}-\tau _{i-1}^{2}\right) +b_{i}^{2}\left( \tau
_{i}-\tau _{i-1}\right)
\end{equation*}%
according to Lemma \ref{l1}, and%
\begin{eqnarray*}
x\left( \tau _{i}\right) &=&x\left( \tau _{i-1}\right) +v\left( \tau
_{i-1}\right) \left( \tau _{i}-\tau _{i-1}\right) \\
&&+\frac{b_{i}}{2}\left( \tau _{i}-\tau _{i-1}\right) ^{2}+\frac{a_{i}}{6}%
\left( \tau _{i}^{3}+2\tau _{i-1}^{3}-3\tau _{i-1}^{2}\tau _{i}\right) \\
v\left( \tau _{i}\right) &=&v\left( \tau _{i-1}\right) +b_{i}\left( \tau
_{i}-\tau _{i-1}\right) +\frac{a_{i}}{2}\left( \tau _{i}^{2}-\tau
_{i-1}^{2}\right) .
\end{eqnarray*}
\end{problem}

Note that we do not include the constraint (\ref{p22}) here since we have
established the result in \cite{meng2018optimal, meng2018optimala} that the optimal
acceleration profile contains either acceleration or deceleration, but not
both. Therefore, the terminal velocity constraint (\ref{p31}) can replace
the velocity constraint (\ref{p21}). Also based on the analysis in \cite%
{meng2018optimal, meng2018optimala}, the initial acceleration constraint (\ref{p32}) is
sufficient instead of using (\ref{p23}) and (\ref{p24}).

\subsection{Multiple Intersections}
\label{fdm} The optimal parametric framework for two intersections can be
easily extended to the case of more than two intersections. We can use three
triplets $\left( a_{i},b_{i},\tau _{i}\right) $ to parameterize the optimal
acceleration profile for a single intersection. For double intersections,
seven triplets $\left( a_{i},b_{i},\tau _{i}\right) $ are enough to
parameterize the optimal acceleration profile. It can be obtained by
mathematical induction that $4\left( N-1\right) +3$ triplets of the form $\left(
a_{i},b_{i},\tau _{i}\right) $ are enough to characterize the optimal
acceleration profile for $N$ intersections, where the proof will be shown in
a later version of this paper.

Therefore, for $N$ intersections, the ECO-AND problem can be solved by the
following optimization problem:

\begin{problem}
\label{P4} ECO-AND problem%
\begin{equation*}
\min \rho _{t}\tau_{4(N-1)+3}+\rho _{u}\sum_{i=1}^{4\left( N-1\right)
+3}J^{u}_{i}
\end{equation*}%
subject to%
\begin{align*}
& v_{\min }\leq v\left( \tau _{i}\right) \leq v_{\max } \\
& \left( a_{i}\tau _{i}+b_{i}\right) \left( a_{i}\tau _{i-1}+b_{i}\right)
\geq 0 \\
& u_{\min }\leq a_{i}\tau _{i}+b_{i}\leq u_{\max }, \\
& u_{\min }\leq a_{i}\tau _{i-1}+b_{i}\leq u_{\max }, \\
& \tau _{i-1}\leq \tau _{i}, \\
& i=1,2,\ldots ,4\left( N-1\right) +3, \\
& k_{\left\lceil \frac{j}{4}\right\rceil}T_{\left\lceil \frac{j}{4}%
\right\rceil}\leq \tau _{j}\leq k_{\left\lceil \frac{j}{4}%
\right\rceil}T_{\left\lceil \frac{j}{4}\right\rceil}+D_{\left\lceil \frac{j}{%
4}\right\rceil}T_{\left\lceil \frac{j}{4}\right\rceil}, \\
& x\left( \tau _{j}\right) =\sum\nolimits_{i=1}^{\left\lceil \frac{j}{4}%
\right\rceil }l_{i} \\
& j=4,8,\ldots ,4\left( N-1\right) ,4\left( N-1\right) +3,
\end{align*}%
where $J^{u}_{i}$ is the energy cost during the interval $\left[ \tau
_{i-1},\tau _{i}\right] $, which can be obtained as%
\begin{equation*}
J^{u}_{i}=\frac{a_{i}^{2}}{3}\left( \tau _{i}^{3}-\tau _{i-1}^{3}\right)
+a_{i}b_{i}\left( \tau _{i}^{2}-\tau _{i-1}^{2}\right) +b_{i}^{2}\left( \tau
_{i}-\tau _{i-1}\right)
\end{equation*}%
from Lemma 1, $\left\lceil x\right\rceil $ is the smallest integer greater
than or equal to $x,$%
\begin{equation*}
v\left( \tau _{i}\right) =v\left( \tau _{i-1}\right) +b_{i}\left( \tau
_{i}-\tau _{i-1}\right) +\frac{a_{i}}{2}\left( \tau _{i}^{2}-\tau
_{i-1}^{2}\right)
\end{equation*}%
and%
\begin{eqnarray*}
x\left( \tau _{i}\right) &=&\sum_{j=1}^{i-1}l_{j}+v\left( \tau _{i-1}\right)
\left( \tau _{i}-\tau _{i-1}\right) \\
&&+\frac{b_{i}}{2}\left( \tau _{i}-\tau _{i-1}\right) ^{2}+\frac{a_{i}}{6}%
\left( \tau _{i}^{3}+2\tau _{i-1}^{3}-3\tau _{i-1}^{2}\tau _{i}\right) .
\end{eqnarray*}
\end{problem}

\section{EXTENSION TO CASES WITH INTERFERING TRAFFIC}
In the above results, we assume that a single vehicle operates in free flow mode. However, the proposed method can be easily extended to traffic conditions where other road users may affect the movement of the autonomous vehicle. Therefore, a safety constraint has to be enforced at all times, that is,%
\begin{equation}
x_{h}\left( t\right)
-x\left( t\right)  \geq \alpha v\left( t\right) +\beta \label{sc}
\end{equation}%
where $x_{h}$ is the position of the preceding vehicle, $\alpha$ and $\beta$ are two scalars representing dynamic and static gaps, respectively \cite{ferrara2018freeway}. The following assumptions are made:
\begin{itemize}
\item On the road, the future speed and acceleration profiles of the preceding vehicle can be estimated accurately enough by the autonomous vehicle.
\item At the intersection, the queue information and stopped vehicle lengths are also available to the autonomous vehicle.
\end{itemize}%
When the proposed approach without considering interfering traffic is applied, the safety constraint in (\ref{sc}) may be violated. The first case is that the preceding vehicle will cross the intersection at some $t \in [kT_{i}, kT_{i}+D_{i}T_{i}]$ while the autonomous vehicle will not. In this case, the autonomous vehicle becomes the leading vehicle on the road, and the intersection crossing time is set as the beginning of the next green light interval $t_{i}=kT_{i+1}$ for some positive integer $k$. The safety constraint may not be violated since the preceding vehicle will accelerate or cruise through the intersection while the autonomous vehicle will decelerate to approach the intersection in the next green light interval. The second case is that both the preceding vehicle and the autonomous vehicle will cross the intersection at the same green light interval. In this case, we can simply decrease the maximum speed to $\theta v_{\max}$ with $0<\theta<1$ so that the safety constraint is satisfied at all times. The third case is that the preceding vehicle stopped before the intersection. In this case, the autonomous vehicle will cross the intersection after the preceding vehicle with a certain time gap $\sigma$ when the traffic signal turns to green so that the safety constraint is not violated, that is, $t_{i}=kT_{i}+\sigma$.

\section{SIMULATION EXAMPLES}
We evaluate the proposed solution by testing the following scenario with two
intersections in MATLAB. The length for each road segment is 200 meters. Each intersection is equipped with a traffic light. Two phases were set up for each signal. The total cycle is 40 seconds, where the green time is set to 20 seconds. The speed
limits are set as $v_{m}=2.78$ $m/s,$ and $v_{M}=20$ $m/s$. The maximum
acceleration and deceleration are $u_{M}=2.5$ $m/s^{2}$ and $u_{m}=-2.9$ $%
m/s^{2}$. We assume that the vehicle starts with $v_{0}=0$. We use our
previous approach for a single intersection \cite{meng2018optimal, meng2018optimala} as the
baseline scenario, which
solves the eco-driving problem for each road segment individually, and compare the proposed solution and the baseline
scenario. Table~\ref{tab11} shows the performance by using our previous
approach \cite{meng2018optimal, meng2018optimala} and the optimal parametric approach. Even
though our previous approach \cite{meng2018optimal, meng2018optimala} calculates the optimal
performance for each road segment, it is not the optimal solution for the
\emph{combined} two segments as a whole. 
Overall,
the optimal multi-intersection parametric approach outperforms \cite{meng2018optimal, meng2018optimala} by $10.29\%$.
Figures \ref{fig}-\ref{fig2} show the acceleration, speed, and distance
profiles for both the optimal multi-intersection parametric approach (blue curve) and \cite%
{meng2018optimal, meng2018optimala} (red curve). The optimal multi-intersection parametric approach verifies the
properties of the optimal acceleration profile, that is, continuity and $%
u^{*}(t_{p}^{*})=0$ even though such constraints are not enforced in Problem~%
\ref{P2}. In addition, the speed profile of the optimal multi-intersection parametric approach is
smoother than that of \cite{meng2018optimal, meng2018optimala} as seen from Fig.~\ref{fig1}.
We can see from Fig.~\ref{fig2} that the intersection crossing times of both
approaches are within the green light interval. The travel time is 40
seconds for both approaches, which is determined by the second traffic
light. However, their first intersection crossing times are different.
\begin{figure}[htbp]
\centerline{\includegraphics[width=\linewidth]{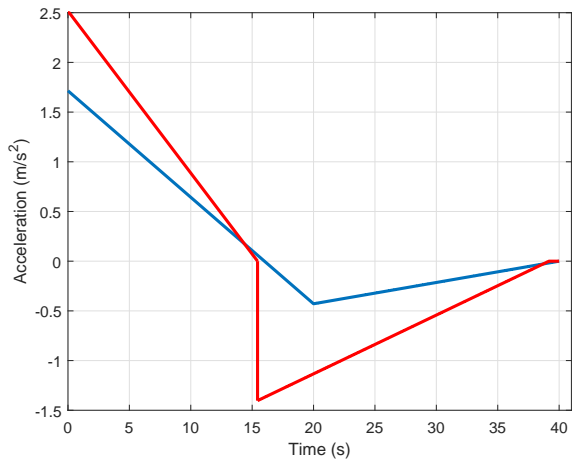}}
\caption{Acceleration profile of different methods: blue is the optimal multi-intersection parametric
approach and red is the controller from \cite{meng2018optimal,meng2018optimala}.}
\label{fig}
\end{figure}

\begin{figure}[htbp]
\centerline{\includegraphics[width=\linewidth]{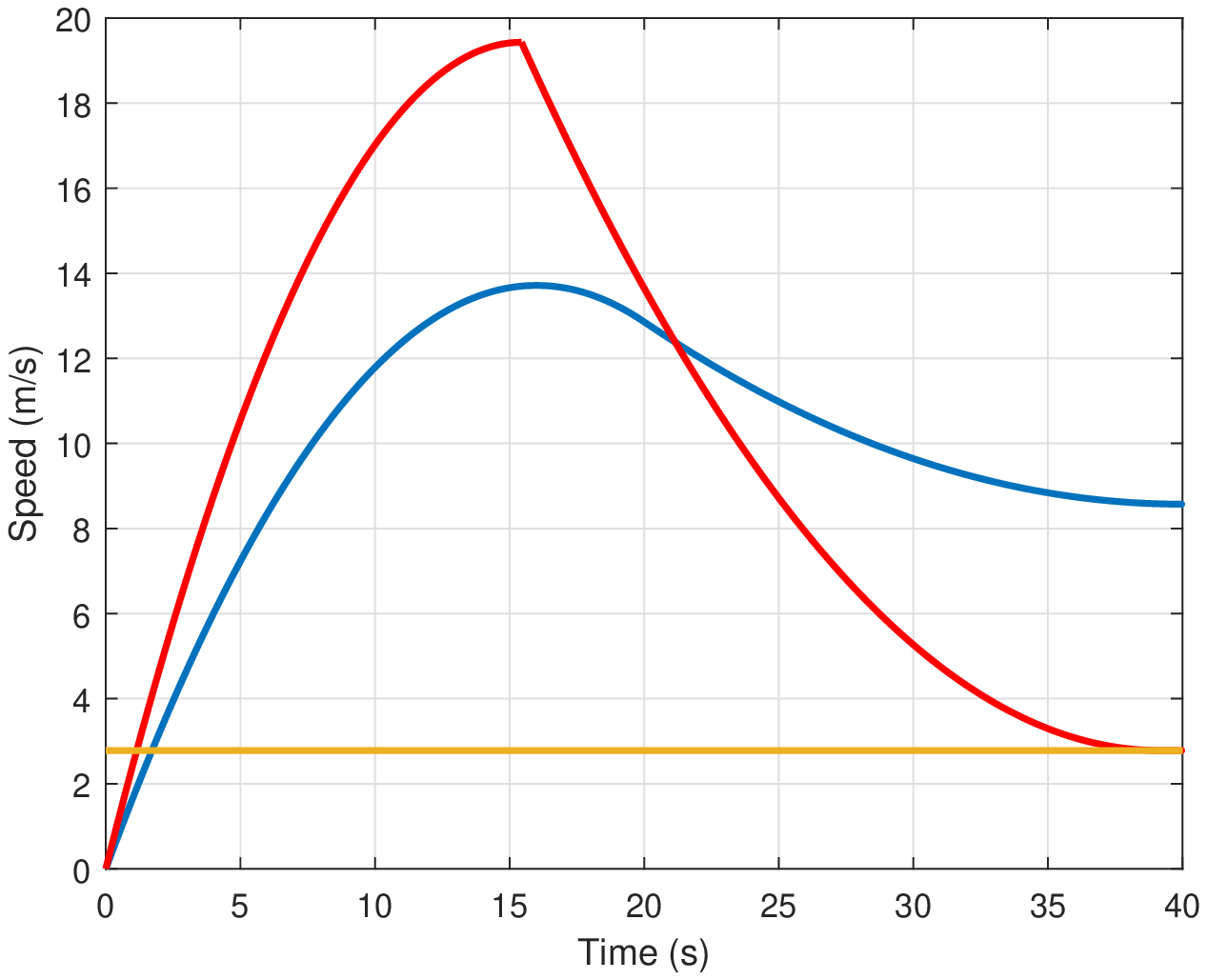}}
\caption{Speed profile of different methods: blue is the optimal multi-intersection parametric approach
and red is the controller from \cite{meng2018optimal,meng2018optimala}.}
\label{fig1}
\end{figure}

\begin{figure}[htbp]
\centerline{\includegraphics[width=\linewidth]{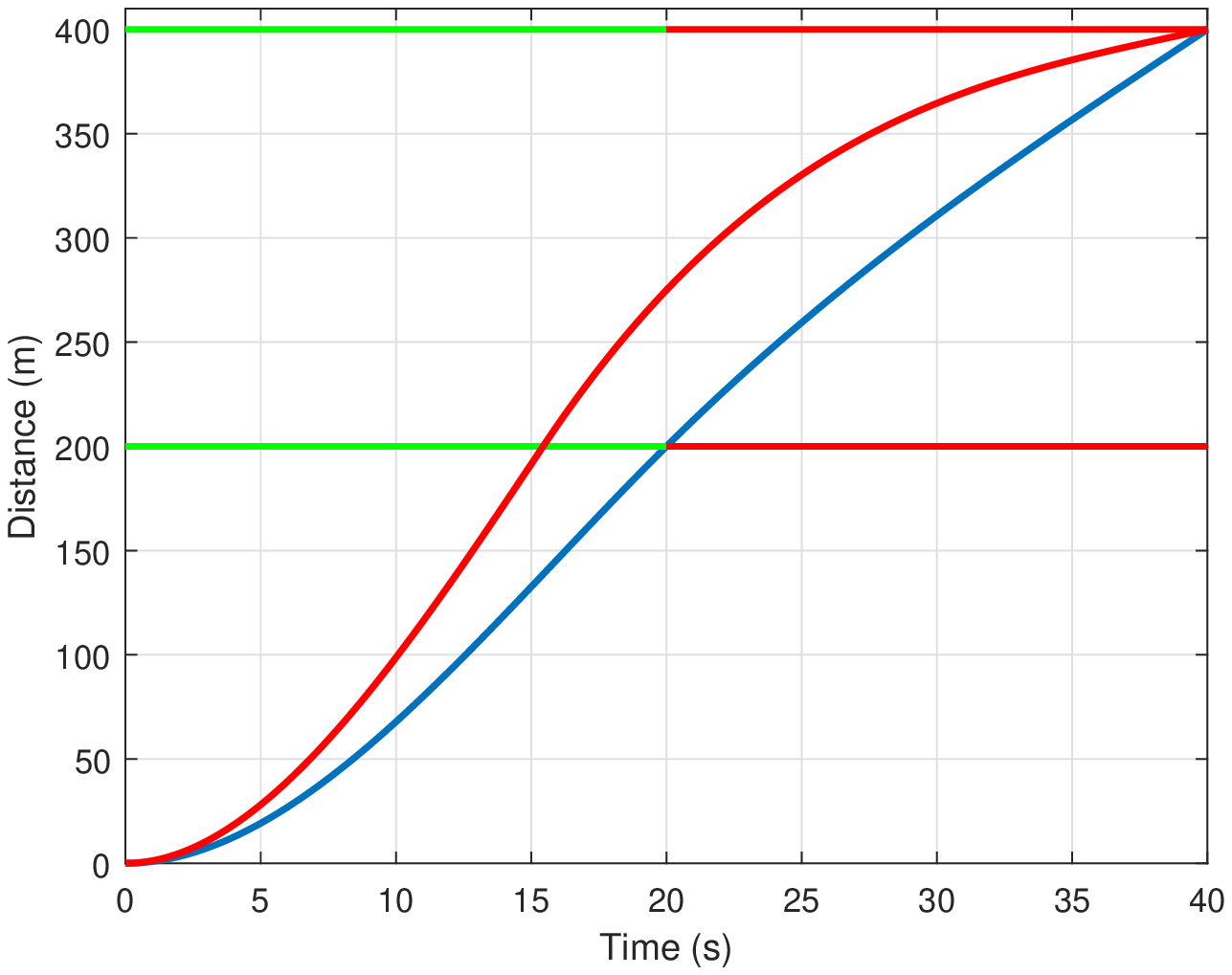}}
\caption{Distance profile of different methods: blue is the optimal multi-intersection parametric
approach and red is the controller from \cite{meng2018optimal,meng2018optimala}.}
\label{fig2}
\end{figure}

\begin{table}[htbp]
\caption{Performance Comparison}
\label{tab11}
\begin{center}
\begin{tabular}{|c|c|c|c|}
\hline
\multirow{2}{*}{\textbf{Method}} & \multicolumn{3}{|c|}{\textbf{Performance}}
\\ \cline{2-4}
& $J_{1}$ & $J_{2}$ & $J$ \\ \hline
\cite{meng2018optimal, meng2018optimala} & $0.1366$ & $0.1793$ & $0.3159$ \\ \hline
Optimal Parametric Approach & $0.1494$ & $0.1340$ & $0.2834$ \\ \hline
Improvement & $-9.67\%$ & $25.26\%$ & $10.29\%$ \\ \hline
\end{tabular}%
\end{center}
\end{table}

\subsection{Example with Interfering Traffic}
In the following, we consider a scenario where the autonomous vehicle will not be obstructed on the road but there is a vehicle which will stop before the second intersection. We assume that such information is available to the autonomous vehicle at time $t_{0}$. It is infeasible for the autonomous vehicle to cross the second intersection at 40 seconds as before. We assume that the feasible intersection crossing time is $t\in [44, 60]$, where the four seconds include driver's reaction time, headway time, and time for the stopped vehicle to clear the intersection. Figures \ref{fig4}-\ref{fig6} depict the acceleration profiles, speed profiles, and distance profiles, respectively, for both the cases with and without interfering traffic. It can be seen from the figures that the first intersection crossing times are the same for both cases. However, when there is interfering traffic, the autonomous vehicle has to start with a higher acceleration and decelerate more before the first intersection compared with the case without interfering traffic.

\begin{figure}[htbp]
\centerline{\includegraphics[width=\linewidth]{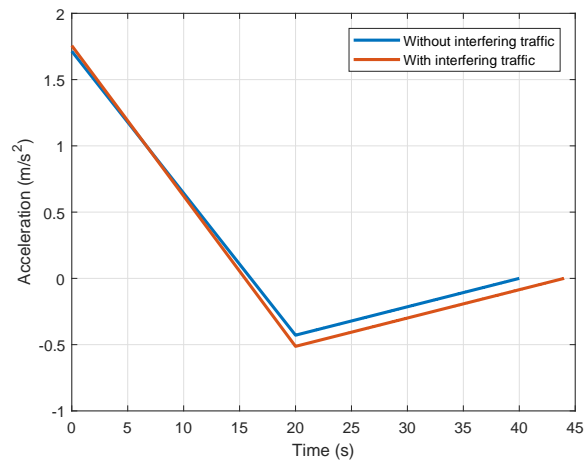}}
\caption{Acceleration profiles of the optimal solution with and without interfering traffic.}
\label{fig4}
\end{figure}

\begin{figure}[htbp]
\centerline{\includegraphics[width=\linewidth]{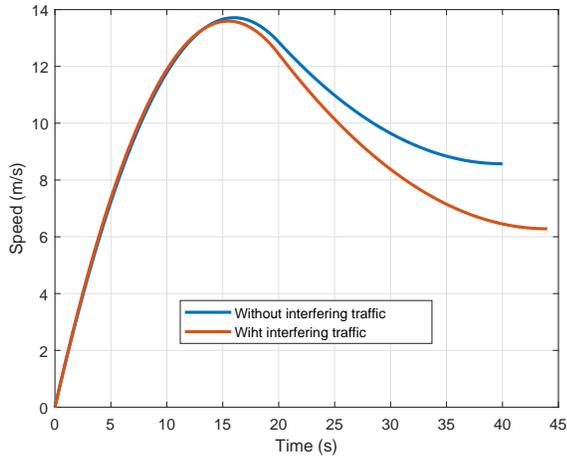}}
\caption{Speed profiles of the optimal solution with and without interfering traffic.}
\label{fig5}
\end{figure}

\begin{figure}[htbp]
\centerline{\includegraphics[width=\linewidth]{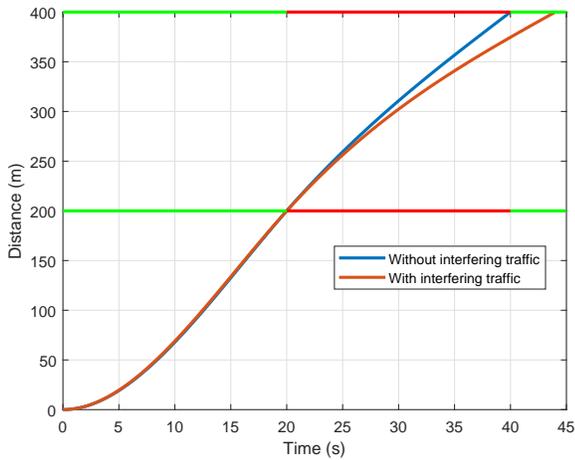}}
\caption{Distance profiles of the optimal solution with and without interfering traffic.}
\label{fig6}
\end{figure}

Here we do not compare the results with those of human-driven vehicles due to space constraints. Such a comparison was done in our previous work for a single intersection \cite{meng2018optimal, meng2018optimala}, where 2\%-10\% performance improvement was shown in terms of travel time and fuel consumption.
\section{Conclusions}

In this paper, we solve an eco-driving problem of
autonomous vehicles crossing multiple intersections without stopping.
Spatial equality constraints and temporal inequality constraints are used to
capture the traffic light constraints. Inspired by our previous work on a
single intersection, the optimal acceleration profile is proved to have a
piece-wise linear parametric structure. We illustrate the effectiveness of
the proposed parametric approach through simulation examples. The results
show that the performance is significantly improved by using the proposed
optimal parametric approach compared with our previous approach which is optimal for
each individual intersection decoupled from the other. We also show that the optimal eco-driving algorithm is capable of dealing with interfering traffic.

\end{document}